\documentclass{article}

\usepackage{amssymb}
\usepackage{graphicx}
\usepackage{amsmath}
\usepackage{latexsym}
\usepackage{amsthm}
\usepackage{enumerate}
\usepackage[ruled,vlined,linesnumbered]{algorithm2e}
\SetArgSty{textnormal}
\SetKwProg{Procedure}{procedure}{}{end}
\SetKwProg{Function}{function}{}{end}

\theoremstyle{plain}
\newtheorem{theorem}{Theorem}
\newtheorem{lemma}{Lemma}

\theoremstyle{definition}
\newtheorem{definition}{Definition}
\newtheorem{example}{Example}

\theoremstyle{remark}

\title{Linear-space LCS enumeration with quadratic-time delay for two strings}

\author{Yoshifumi Sakai \thanks{Graduate School of Agricultural Science, Tohoku University, Japan}}

\begin{document}

\maketitle

\begin{abstract}
Suppose we want to seek the longest common subsequences (LCSs) of two strings as informative patterns that explain the relationship between the strings.
The dynamic programming algorithm gives us a table from which all LCSs can be extracted by traceback.
However, the need for quadratic space to hold this table can be an obstacle when dealing with long strings.
A question that naturally arises in this situation would be whether it is possible to exhaustively search for all LCSs one by one in a time-efficient manner using only a space linear in the LCS length, where we treat read-only memory for storing the strings as excluded from the space consumed.
As a part of the answer to this question, we propose an $O(L)$-space algorithm that outputs all distinct LCSs of the strings one by one each in $O(n^2)$ time, where the strings are both of length $n$ and $L$ is the LCS length of the strings.
\end{abstract}

\section{Introduction}\label{sec introduction}

Comparing two strings to find common patterns that would be most informative of their relationships is a fundamental task in parsing them.
The longest common subsequences (LCSs) are regarded as such common patterns, and the problem of efficiently finding one of them has long been investigated~\cite{Apo,AG,CP,GH,Hir,HS,IR,MP,NKY,WF}.
Here, a subsequence of a string is the sequence obtained from the string by deleting any number of elements at any position that is not necessarily contiguous.
Hence, an LCS of two strings is a common subsequence obtained by deleting the least possible number of elements from the strings.
For example, if $X = \mathtt{acddadacbcb}$ and $Y = \mathtt{caccbaadcad}$ are the target strings, then $\mathtt{caccb}$ is one of the seven distinct LCSs of $X$ and $Y$.
The six other LCSs are $\mathtt{cacbc}$, $\mathtt{accbc}$, $\mathtt{acaac}$, $\mathtt{acadc}$, $\mathtt{acada}$, and $\mathtt{acdad}$.

Given a pair of strings $X$ and $Y$ both of length $n$ and a string $Z$ of length less than $n$, it is easy to determine whether $Z$ is a common subsequence of $X$ and $Y$ in $O(n)$ time.
In comparison, it is hard to determine whether $Z$ is an LCS of $X$ and $Y$.
The reason is that we need to know the LCS length in advance.
It was revealed~\cite{ABW,BK} that unless the strong exponential time hypothesis (SETH) does not hold, no algorithm can determine the LCS length of $X$ and $Y$ in $O(n^{2 - \varepsilon})$ time for any positive constant $\varepsilon$.
Therefore, we cannot find any LCS in $O(n^{2 - \varepsilon})$ time under SETH.

On the other hand, as is well known, finding an arbitrary LCS of $X$ and $Y$ is possible in $O(n^2)$ time and $O(n^2)$ space by the dynamic programming (DP) algorithm~\cite{WF}.
The space of size $O(n^2)$ consumed to store the LCS lengths for all pairs of a prefix of $X$ and a prefix of $Y$ can be treated as constituting a particular directed acyclic graph (DAG) such that any path from the source to the sink represents an LCS of $X$ and $Y$.
This DAG is also available in applications other than seeking a single arbitrary LCS because each existing LCS corresponds to at least one of the paths from the source to the sink.
If one considers only the problem of finding a single arbitrary LCS without intending a data structure representing all LCSs, Hirschberg~\cite{Hir} showed that $O(n)$ space is sufficient to solve the problem in the same $O(n^2)$ time as the DP algorithm.
Excluding the read-only memory storing $X$ and $Y$, his algorithm performs only in $O(L)$ space with a slight modification, where $L$ is the LCS length of $X$ and $Y$.

The advantage of $O(L)$-space algorithms for finding an arbitrary LCS is that they perform without reserving $O(n^2)$ space, the size of which becomes pronounced when $n$ is large.
Thus, a natural question that would come to mind is whether it is possible to design $O(L)$-space algorithms that are as accessible to all LCSs as the DP algorithm in addition to the above advantage.
As a part of the answer to this question, the author showed in \cite{Sak} that an $O(L)$-space algorithm is capable of enumerating all distinct LCSs with $O(n^2 \log L)$-time delay, which is asymptotically inferior to the execution time of Hirschberg's $O(L)$-space algorithm~\cite{Hir} for finding a single arbitrary LCS by a factor $O(\log L)$.
Enumeration is done by introducing a DAG, called the all-LCS graph, with each path from the source to the sink corresponding to a distinct LCS and vice versa.
Since the size of the all-LCS graph is $O(n^2)$, which exceeds $O(L)$, the proposed algorithm performs a depth-first search without explicitly constructing it.

In this article, we consider the same enumeration problem as above.
Our contribution is to address the shortcomings in the previous linear-space LCS enumeration algorithm~\cite{Sak} concerning delay time.
We accomplished this by replacing the crucial part of the algorithm, which finds the next LCS strictly specified by the immediately preceding output LCS, with a variant of Hirschberg's linear-space LCS-finding algorithm~\cite{Hir}.
Similarly to Hirschberg's algorithm, the crucial part of the previous algorithm finds the LCS by divide and conquer.
The reason why the delay time of the algorithm is longer than the execution time of Hirschberg's algorithm is that Hirschberg's algorithm finds a balanced midpoint concerning the length of input strings when dividing them.
In contrast, the previous algorithm for LCS enumeration could only find a balanced midpoint concerning the length of the LCS it seeks.
The variant of Hirschberg's algorithm that we adopt to improve delay time differs from the original in that the chosen midpoint is strictly specified.
This difference would be natural since the original can find an arbitrary LCS, whereas our algorithm must find a specified LCS.
We show that the exact choice of the midpoint in Hirschberg's algorithm can improve the delay time of the previous linear-space LCS-enumeration algorithm~\cite{Sak} from $O(n^2 \log L)$ to $O(n^2)$.

\section{Preliminaries}\label{sec pre}

For any sequences $S$ and $S'$, let $S \circ S'$ denote the concatenation of $S$ followed by $S'$.
For any sequence $S$, let $|S|$ denote the length of $S$, i.e., the number of elements in $S$.
For any index $i$ with $1 \leq i \leq |S|$, let $S[i]$ denote the $i$th element of $S$, so that $S = S[1] \circ S[2] \circ \cdots \circ S[|S|]$.
A \emph{subsequence} of $S$ is the sequence obtained from $S$ by deleting zero or more elements at any positions not necessarily contiguous, i.e., the sequence $S[i_1] \circ S[i_2] \circ \cdots \circ S[i_\ell]$ for some length $\ell$ with $0 \leq \ell \leq |S|$ and some $\ell$ indices $i_1,i_2,\dots,i_\ell$ with $1 \leq i_1 < i_2 < \cdots < i_\ell \leq |S|$.
We sometimes denote this subsequence of $S$ by $S[i_1 \circ i_2 \circ \cdots \circ i_\ell]$, or $S[I]$ with $I = i_1 \circ i_2 \circ \cdots \circ i_\ell$.
For any indices $i$ and $i'$ with $1 \leq i \leq i' + 1 \leq |S| + 1$, let $S[i : i']$ denote the contiguous subsequence $S[i] \circ S[i + 1] \circ \cdots \circ S[i']$ of $S$, where $S[i : i']$ with $i = i' + 1$ represents the empty subsequence.
If $i = 1$ (resp. $i' = |S|$), then $S[i : i']$ is called a \emph{prefix} (resp. \emph{suffix}) of $S$.

A string is a sequence of characters over an alphabet set.
For any strings $X$ and $Y$, a \emph{common subsequence} of $X$ and $Y$ is a subsequence of $X$ that is a subsequence of $Y$.
Let $L(X, Y)$ denote the maximum of $|Z|$ over all common subsequences $Z$ of $X$ and $Y$.
Any common subsequence $Z$ of $X$ and $Y$ with $|Z| = L(X, Y)$ is called a \emph{longest common subsequence} (an \emph{LCS}) of $X$ and $Y$.
The following lemma gives a typical recursive expression for the LCS length.

\begin{lemma}[e.g., \cite{WF}]\label{lem L}
For any strings $X$ and $Y$, if at least one of $X$ and $Y$ is empty, then $L(X, Y) = 0$; otherwise, if $x = y$, then $L(X, Y) = L(X', Y') + 1$; otherwise, $L(X, Y) = \max(L(X', Y), L(X, Y'))$, where $X = X' \circ x$ and $Y = Y' \circ y$.
The same also holds for the case where $X = x \circ X'$ and $Y = y \circ Y'$.
\end{lemma}

Let any algorithm that takes an arbitrary pair of non-empty strings $X$ and $Y$ as input and uses only $O(L(X, Y))$ space to output all distinct LCSs of $X$ and $Y$ one by one each in $O(|X||Y|)$ time be called a \emph{linear-space LCS enumeration algorithm with quadratic-time delay}.
To make this definition meaningful under the space constraint described above, we assume that $X$ and $Y$ are not directly given to the algorithm, but rather $|X|$, $|Y|$ and $O(1)$-time access to whether $X[i] = Y[j]$ for any index pair $(i, j)$ with $1 \leq i \leq |X|$ and $1 \leq j \leq |Y|$ are given without space consumption.
In addition, we assume that the algorithm outputs, instead of each distinct LCS $Z$ of $X$ and $Y$, a sequence $P$ of $L(X, Y)$ increasing indices between $1$ and $Y$ such that $Y[P] = Z$.

\section{Algorithm}\label{sec algo}

This section proposes a linear-space LCS enumeration algorithm with quadratic-time delay.

To handle each distinct LCS $Z$, we use the ``leftmost'' one as the representative of the sequences of $L(X, Y)$ increasing indices that represent the positions at which $Z$ appears in $Y$.
The proposed algorithm outputs all such representative sequences one by one in lexicographic order.

\begin{definition}[$\mathcal{P}(\tilde{X}, \tilde{Y})$]\label{def P}
For any contiguous subsequences $\tilde{X}$ of $X$ and $\tilde{Y} = Y[j' : j'']$ of $Y$, let any sequence $P$ of $L(\tilde{X}, \tilde{Y})$ increasing indices between $j'$ and $j''$ such that $Y[P]$ is an LCS of $\tilde{X}$ and $\tilde{Y}$ and for any index $k$ with $1 \leq k \leq L(\tilde{X}, \tilde{Y})$, $Y[j' : P[k]]$ is the shortest prefix of $\tilde{Y}$ with $Y[P[1 : k]]$ as a subsequence be called an \emph{LCS-position sequence of $\tilde{X}$ and $\tilde{Y}$}.
Let $\mathcal{P}(\tilde{X}, \tilde{Y})$ denote the sequence of all LCS-position sequences of $\tilde{X}$ and $\tilde{Y}$ in lexicographic order.
\end{definition}

\begin{example}\label{ex P}
If $X = \mathtt{acddadacbcb}$ and $Y = \mathtt{caccbaadcad}$ (with $\circ$ omitted), then $\mathcal{P}(X, Y)$ consists of the seven LCS-position sequences listed below in the same order.
\begin{itemize}
\item $1 \circ 2 \circ 3 \circ 4 \circ 5$ (representing $\mathtt{caccb}$)
\item $1 \circ 2 \circ 3 \circ 5 \circ 9$ (representing $\mathtt{cacbc}$)
\item $2 \circ 3 \circ 4 \circ 5 \circ 9$ (representing $\mathtt{accbc}$)
\item $2 \circ 3 \circ 6 \circ 7 \circ 9$ (representing $\mathtt{acaac}$)
\item $2 \circ 3 \circ 6 \circ 8 \circ 9$ (representing $\mathtt{acadc}$)
\item $2 \circ 3 \circ 6 \circ 8 \circ 10$ (representing $\mathtt{acada}$)
\item $2 \circ 3 \circ 8 \circ 10 \circ 11$ (representing $\mathtt{acdad}$)
\end{itemize}
\end{example}

The algorithm searches for the next element $P^\star$ in $\mathcal{P}(X, Y)$ of the last found element $P$ after determining somehow the features of $P$ introduced below.

\begin{definition}[$k^\star_P$ and $j^\star_P$]\label{def branch}
For any consecutive elements $P$ and $P^\star$ after $P$ in $\mathcal{P}(X, Y)$, let $k^\star_P$ denote the least index $k$ with $P[k] < P^\star[k]$ and let $j^\star_P$ denote $P^\star[k^\star_P]$, so that $P^\star[1 : k^\star_P] = P[1 : k^\star_P - 1] \circ j^\star_P$.
\end{definition}

\begin{lemma}\label{lem branch}
For any consecutive elements $P$ and $P^\star$ after $P$ in $\mathcal{P}(X, Y)$, $P^\star$ is the concatenation of $P[1 : k^\star_P - 1] \circ j^\star_P$ followed by the first element of $\mathcal{P}(\tilde{X}, \tilde{Y})$, where $\tilde{X}$ (resp. $\tilde{Y}$) is the suffix of $X$ (resp. $Y$) obtained by deleting the shortest prefix with $Y[P[1 : k^\star_P - 1] \circ j^\star_P]$ as a subsequence.
\end{lemma}

\begin{proof}
Since $P^\star[k^\star_P + 1 : L(X, Y)]$ is an element of $\mathcal{P}(\tilde{X}, \tilde{Y})$, the lemma follows from the fact that $\mathcal{P}(X, Y)$ has no element between $P$ and $P^\star$.
\end{proof}

\begin{algorithm}[t]
\DontPrintSemicolon
$\mathsf{P} \leftarrow$ the empty sequence; $\mathsf{k}^\star \leftarrow 0$\;
\Repeat{$\mathsf{k}^\star$ is a dummy index}{
  $\mathsf{P} \leftarrow \mathsf{P}[1 : \mathsf{k}^\star] \circ \textsc{firstLCS}(i + 1, |X|, j + 1, |Y|)$, where $X[1 : i]$ (resp. $Y[1 : j]$) is the shortest prefix of $X$ (resp. $Y$) with $Y[\mathsf{P}[1 : \mathsf{k}^\star]]$ as a subsequence\;
  output $\mathsf{P}$\;
  \textsc{findBranch}
}
\caption{\textsc{enumerateLCS}\label{algo enumerateLCS}}
\end{algorithm}

The proposed algorithm, which we denote \textsc{enumerateLCS}, finds each element $P$ in $\mathcal{P}(X, Y)$ as a distinct LCS-position sequence based on Lemma~\ref{lem branch} using an index sequence variable $\mathsf{P}$ of length $L(X, Y)$ and an index variable $\mathsf{k}^\star$.
At the beginning of each process corresponding to $P$, these variables are maintained so that $\mathsf{P}[1 : \mathsf{k}^\star] = P[1 : \mathsf{k}^\star]$ and $P[\mathsf{k}^\star + 1 : L(X, Y)]$ is the first element in $\mathcal{P}(\tilde{X}, \tilde{Y})$, where $\tilde{X}$ (resp. $\tilde{Y}$) is the suffix of $X$ (resp. $Y$) obtained by deleting the shortest prefix with $Y[P[1 : \mathsf{k}^\star]]$ as a subsequence.
The algorithm finds the first element $\tilde{P}$ of $\mathcal{P}(\tilde{X}, \tilde{Y})$ by calling function $\textsc{firstLCS}$, as which we adopt a variant of Hirschberg's linear-space LCS-finding algorithm~\cite{Hir} with a very slight modification.
After updating $\mathsf{P}$ to $\mathsf{P}[1 : \mathsf{k}^\star] \circ \tilde{P}$ and outputting it as $P$, the algorithm determines the existing $k^\star_P$ and $j^\star_P$ by executing procedure \textsc{findBranch}, which sets $\mathsf{k}^\star$ and $\mathsf{P}[\mathsf{k}^\star]$ to $k^\star_P$ and $j^\star_P$, respectively, if existing, or sets $\mathsf{k}^\star$ to a dummy index to indicate that $P$ is the last element of $\mathcal{P}(X, Y)$, otherwise.
Consequently, the condition that should be satisfied at the beginning of the process for the next element in $\mathcal{P}(X, Y)$ after $P$ holds due to Lemma~\ref{lem branch}.
Hence, induction proves that all elements in $\mathcal{P}(X, Y)$ are output one by one in lexicographic order.
A pseudo-code of \textsc{enumerateLCS} is given as Algorithm~\ref{algo enumerateLCS}.

Sections~\ref{sec firstLCS} and \ref{sec findBranch} respectively design \textsc{firstLCS} and \textsc{findBranch} to perform as mentioned above in $O(|X||Y|)$ time and $O(L(X,Y))$ space, which prove the following theorem.

\begin{theorem}\label{theo enumerateLCS}
Algorithm \textsc{enumerateLCS} works as a linear-space LCS enumeration algorithm with quadratic-time delay.
\end{theorem}

\begin{proof}
The theorem follows from Lemma~\ref{lem firstLCS} in Section~\ref{sec firstLCS} and Lemma~\ref{lem findBranch} in Section~\ref{sec findBranch}.
\end{proof}

\subsection{Function \textsc{firstLCS}}\label{sec firstLCS}

\begin{algorithm}[t]
\DontPrintSemicolon
\eIf{$i' = i''$}{
  \ForEach{$j$ from $j'$ to $j''$}{
    \lIf{$X[i'] = Y[j]$}{\Return $j$}
  }
  \Return the empty sequence
}{
  $J^-, J^+ \leftarrow$ the empty sequences; $\bar\imath \leftarrow \lfloor (i' + i'')/2 \rfloor$\;
  \ForEach{$i$ from $i'$ to $\bar\imath$}{
    $l \leftarrow |J^-|$\;
    \ForEach{$j$ from $j''$ down to $j'$}{
      \lIf{$l > 0$ and $J^-[l] = j$}{$l^- \leftarrow l - 1$}
      \lIf{$X[i] = Y[j]$}{$J^-[l + 1] \leftarrow j$}
    }
  }
  \ForEach{$i$ from $i''$ down to $\bar\imath + 1$}{
    $l \leftarrow |J^+|$\;
    \ForEach{$j$ from $j'$ to $j''$}{
      \lIf{$l > 0$ and $J^+[l] = j$}{$q \leftarrow l - 1$}
      \lIf{$X[i] = Y[j]$}{$J^+[l + 1] \leftarrow j$}
    }
  }
  $l^- \leftarrow 0$; $l^+ \leftarrow |J^+|$; $l \leftarrow l^- + l^+$; $\bar\jmath \leftarrow j' - 1$\;
  \ForEach{$j$ from $j'$ to $j''$}{
    \lIf{$J^-[l^- + 1] \leq j$}{$l^- \leftarrow l^- + 1$}
    \lIf{$J^+[l^+] \leq j$}{$l^+ \leftarrow l^+ - 1$}
    \If{$l^- + l^+ > l$}{
      $\bar\jmath \leftarrow j$; $l \leftarrow l^- + l^+$
    }
  }
  \Return $\textsc{firstLCS}(i', \bar\imath, j', \bar\jmath) \circ \textsc{firstLCS}(\bar\imath + 1, i'', \bar\jmath + 1, j'')$
}
\caption{$\textsc{firstLCS}(i', i'', j', j'')$\label{algo firstLCS}}
\end{algorithm}

This section presents $\textsc{firstLCS}(i', i'', j', j'')$ as a recursive function that returns the first element of $\mathcal{P}(\tilde{X}, \tilde{Y})$ in $O(|\tilde{X}||\tilde{Y}|)$ time and $O(L(\tilde{X}, \tilde{Y})$ space, where $\tilde{X} = X[i' : i'']$ and $Y[j' : j'']$.
As mentioned earlier, we adopt as it a variant of Hirschberg's linear-space LCS-finding algorithm~\cite{Hir}.
A pseudo-code of $\textsc{firstLCS}(i', i'', j', j'')$ is given as Algorithm~\ref{algo firstLCS}.

Following Hirschberg's divide-and-conquer rule, when the length of $\tilde{X}$ is more than one, \textsc{firstLCS} halves $\tilde{X}$ into $X' \circ X''$ with $|X'| = \lceil |\tilde{X}|/2 \rceil$ and finds a division of $\tilde{Y}$ into $Y' \circ Y''$ such that $L(X', Y') + L(X'', Y'') = L(\tilde{X}, \tilde{Y})$.
This division of $\tilde{Y}$ allows \textsc{firstLCS} to return $P' \circ P''$ as an LCS of $\tilde{X}$ and $\tilde{Y}$, where $P'$ (resp. $P''$) is the LCS of $X'$ and $Y'$ (resp. $X''$ and $Y''$) returned by the recursive call of itself with argument representing $X'$ and $Y'$ (resp. $X''$ and $Y''$).
A notable characteristic of \textsc{firstLCS} compared to the original Hirschberg's algorithm is that \textsc{firstLCS} adopts the shortest possible one as $Y'$, when dividing $\tilde{Y}$ into $Y' \circ Y''$.
The following lemma guarantees this characteristic.

\begin{lemma}\label{lem midpoint}
For any pair of $X[i' : i'']$ and $Y[j' : j'']$, if $i' < i''$, then $\bar{\jmath}$ is the least index such that $L(X', Y') + L(X'', Y'') = L(\tilde{X}, \tilde{Y})$, where $\tilde{X} = X[i' : i'']$, $\tilde{Y} = Y[j' : j'']$, $X' = X[i' : \bar{\imath}]$, $Y' = Y[j' : \bar{\jmath}]$, $X'' = X[\bar{\imath} + 1 : i'']$, and $Y'' = Y[\bar{\jmath} + 1 : j'']$ with $\bar{\imath}$ and $\bar{\jmath}$ at line 23 of $\textsc{firstLCS}(i', i'', j', j'')$ implemented as Algorithm~\ref{algo firstLCS}, .
\end{lemma}

\begin{proof}
Since $i' < i''$, \textsc{firstLCS} executes lines 6 through 23.
Lines 6 through 16 construct sequences $J^-$ and $J^+$ such that $|J^-| = L(X', \tilde{Y})$ (resp. $|J^+| = L(X'', \tilde{Y})$) and for any index $p$ (resp. $q$) with $1 \leq p \leq |J^-|$ (resp. $1 \leq q \leq |J^+|$), $J^-[p]$ (resp. $J^+[q])$ is the least (resp. greatest) index $j$ with $L(X', Y[j' : j]) = p$ (resp. $L(X'', Y[j : j'']) = q)$ in a straightforward way based on Lemma~\ref{lem L}.
Therefore, induction proves that at line 21 in each iteration of lines 19 through 22, both $l^- = L(X', Y[i' : j])$ and $l^- = L(X'', Y[j + 1 : j''])$.
Since the maximum of $L(X', Y[i' : j]) + L(X'', Y[j + 1 : j''])$ over all indices $j$ with $j' - 1 \leq j \leq j''$ is equal to $L(\tilde{X}, \tilde{Y})$, the lemma holds due to lines 21 and 22.
\end{proof}

The reason for adopting the shortest $Y'$ strictly to divide $\tilde{Y}$ into $Y' \circ Y''$ is that the first element in $\mathcal{P}(\tilde{X}, \tilde{Y})$ is appropriately inherited from parent to child in the recursion tree of \textsc{firstLCS}.
To show this we need the following lemma.

\begin{lemma}\label{lem non-decreasing}
For any pair of contiguous subsequences $\tilde{X}$ of $X$ and $\tilde{Y}$ of $Y$, any element $P$ in $\mathcal{P}(\tilde{X}, \tilde{Y})$, and any index $k$ with $1 \leq k \leq L(\tilde{X}, \tilde{Y})$, $\tilde{P}[k] \leq P[k]$, where $\tilde{P}$ is the first element in $\mathcal{P}(\tilde{X}, \tilde{Y})$.
\end{lemma}

\begin{proof}
Assume for contradiction that $k$ is the least index with $\tilde{P}[k] > P[k]$ and let $X'$ be the shortest prefix of $\tilde{X}$ with $Y[\tilde{P}[1 : k]]$ as a subsequence.
If $Y[P[1 : k] \circ \tilde{P}[k]]$ is a subsequence of $X'$, then $Y[P[1 : k] \circ \tilde{P}[k : L(\tilde{X}, \tilde{Y})]]$ is a common subsequence of $\tilde{X}$ and $\tilde{Y}$ of length $L(\tilde{X}, \tilde{Y}) + 1$, a contradiction; otherwise, $\tilde{P}[1 : k - 1] \circ P[k : L(\tilde{X}, \tilde{Y})]$ is an element in $\mathcal{P}(\tilde{X}, \tilde{Y})$ before $\tilde{P}$, also a contradiction.
\end{proof}

Lemma~\ref{lem non-decreasing} provides the following property of the recursion tree of \textsc{firstLCS}.

\begin{lemma}\label{lem inheritance}
For any pair of $X[i' : i'']$ and $Y[j' : j'']$, $\tilde{P} = P' \circ P''$, where $\tilde{X}$, $\tilde{Y}$, $X'$, $Y'$, $X''$, and $Y''$ are the same as in Lemma~\ref{lem midpoint} and $\tilde{P}$, $P'$, and $P''$ are the first elements of $\mathcal{P}(\tilde{X}, \tilde{Y})$, $\mathcal{P}(X', Y')$, and $\mathcal{P}(X'', Y'')$, respectively.
\end{lemma}

\begin{proof}
Let $\bar{k}$ be the least index $k$ such that $Y[\tilde{P}[k + 1 : L(\tilde{X}, \tilde{Y})]]$ is a subsequence of $X''$, so that $Y[\tilde{P}[1 : \bar{k}]]$ is a subsequence of $X'$.
Hence, $L(X', Y[j' : \tilde{P}[\bar{k}]]) = \bar{k}$ and $L(X'', Y[\tilde{P}[\bar{k}] + 1 : j'']) = L(\tilde{X}, \tilde{Y}) - \bar{k}$.
For any index $j$ with $j' - 1 \leq j \leq \tilde{P}[\bar{k}] - 1$, if $L(X', Y[j' : j]) = \bar{k}$, then $P \circ \tilde{P}[\bar{k} + 1 : L(\tilde{X}, \tilde{Y})]$ is an element in $\mathcal{P}(\tilde{X}, \tilde{Y})$, where $P$ is the first element in $\mathcal{P}(X', Y[j' : j])$.
This contradicts Lemma~\ref{lem non-decreasing} because $P[\bar{k}] < \tilde{P}[\bar{k}]$.
Thus, $L(X', Y[j' : j]) + L(X'', Y[j + 1 : j'']) < L(\tilde{X}, \tilde{Y})$, implying that $Y' = Y[j' : \tilde{P}[\bar{k}]]$ and $Y'' = Y[\tilde{P}[\bar{k}] + 1 : j'']$.
Since $\tilde{P}[1 : \bar{k}]$ is an element in $\mathcal{P}(X', Y')$, if $\tilde{P}[1 : \bar{k}] \neq P'$, then $P' \circ \tilde{P}[\bar{k} + 1 : L(\tilde{X}, \tilde{Y})]$ appears in $\mathcal{P}(\tilde{X}, \tilde{Y})$ before $\tilde{P}$, a contradiction.
Thus, $P' = \tilde{P}[1 : \bar{k}]$.
By symmetry, $P'' = \tilde{P}[\bar{k} + 1 : L(\tilde{X}, \tilde{Y})]$.
\end{proof}

Due to the property of the recursion tree guaranteed by Lemma~\ref{lem inheritance}, together with the time and space efficiency of the original Hirschberg's algorithm~\cite{Hir}, \textsc{firstLCS} works as desired.

\begin{lemma}\label{lem firstLCS}
For any pair of $\tilde{X} = X[i' : i'']$ and $\tilde{Y} = Y[j' : j'']$, $\textsc{firstLCS}(i', i'', j', j'')$ implemented as Algorithm~\ref{algo firstLCS} returns the first element of $\mathcal{P}(\tilde{X}, \tilde{Y})$ in $O(|\tilde{X}||\tilde{Y}|)$ time and $O(L(\tilde{X}, \tilde{Y}))$ space.
\end{lemma}

\begin{proof}
Induction proves that $\textsc{firstLCS}(i', i'', j', j'')$ returns the first element of $\mathcal{P}(\tilde{X}, \tilde{Y})$ based on Lemma~\ref{lem inheritance} and the fact that if $i' = i''$, then lines 2 through 4 of Algorithm~\ref{algo firstLCS} return the first element of $\mathcal{P}(\tilde{X}, \tilde{Y})$.
By an  argument analogous to the original Hirschberg's algorithm~\cite{Hir}, the time and space required to perform the recursive execution of $\textsc{firstLCS}(i', i'', j', j'')$ are as in the lemma.
\end{proof}

\subsection{Procedure \textsc{findBranch}}\label{sec findBranch}

For any element $P$ in $\mathcal{P}(X, Y)$, \textsc{findBranch} determines existing $k^\star_P$ and $j^\star_P$ without knowing the element $P^\star$ in $\mathcal{P}(X, Y)$ just after $P$ in advance.
To design \textsc{findBranc} as such a procedure, we use the following relationship between elements in $\mathcal{P}(X, Y)$ that share a prefix.

\begin{lemma}\label{lem next}
For any element $P$ in $\mathcal{P}(X, Y)$, if there exist indices $k$ with $1 \leq k \leq L(X, Y)$ and $j^\star$ with $P[k] + 1 \leq j^\star \leq |Y|$ such that $L(X[Q[k - 1] + 1 : i^\star], Y[P[k - 1] + 1 : j^\star] = 1$ and $L(X[i^\star + 1 : |X|], Y[j^\star + 1 : |Y|]) = L(X, Y) - k$, where $X[1 : Q[k - 1]]$ (resp. $X[1 : i^\star]$) is the shortest prefix of $X$ witsh $Y[P[1 : k - 1]]$ (resp. $Y[P[1 : k - 1] \circ j^\star]$) as a subsequence, then $k^\star_P$ is the greatest such $k$ and $j^\star_P$ is the least such $j^\star$ with $k = k^\star_P$; otherwise, $P$ is the last element in $\mathcal{P}(X, Y)$.
\end{lemma}

\begin{proof}
The lemma follows from Definitions~\ref{def P} and \ref{def branch}.
\end{proof}

\begin{algorithm}[t]
\DontPrintSemicolon
$Q \leftarrow$ the empty sequence\;
\ForEach{$i$ from $1$ to $|X|$}{
\If{$|Q| < |\mathsf{P}|$ and $X[i] = Y[\mathsf{P}[|Q| + 1]]$}{
append $i$ to $Q$
}
}
$i^\star \leftarrow |X|$; $J \leftarrow$ the empty sequence; $\mathsf{k}^\star \leftarrow$ a dummy index\;
\ForEach{$k$ from $|\mathsf{P}|$ down to $1$}{
\lWhile{$i^\star \geq Q[k]$}{$\textsc{decI}$}
\ForEach{$j^\star$ from $\mathsf{P}[k] + 1$ to $|Y|$}{
$i \leftarrow Q[k - 1]$\;
\lRepeat{$i = i^\star + 1$ or $X[i] = Y[j^\star]$}{$i \leftarrow i + 1$}
\If{$i \leq i^\star$}{
\lWhile{$i^\star \geq i + 1$}{\textsc{decI}}
\eIf{$|J| \geq |\mathsf{P}| - k$ and $j^\star + 1 \leq J[|\mathsf{P}| - k]$}{
$\mathsf{k}^\star \leftarrow k$; $\mathsf{P}[\mathsf{k}^\star] \leftarrow j^\star$; halt
}{
\textsc{decI}
}
}
}
}
\Procedure{$\textsc{decI}$}{
$i^\star \leftarrow i^\star - 1$; $l \leftarrow |J|$\;
\ForEach{$j$ from $1$ to $|Y|$}{
\lIf{$l > 0$ and $J[l] = j$}{$l \leftarrow l - 1$}
\lIf{$X[i^\star + 1] = Y[j]$}{$J[l + 1] \leftarrow j$}
}
}
\caption{\textsc{findBranch}\label{algo findBranch}}
\end{algorithm}

Based on Lemma~\ref{lem next}, \textsc{findBranch} examines whether $k$ satisfies the condition of $k^\star_P$ for each index $k$ from $L(X, Y)$ to $1$ in descending order.
This is done by using variables $j^\star$ and $i^\star$ as follows.
Let $Q$ be the sequence of $L(X, Y)$ indices such that $X[1 : Q[k]]$ is the shortest prefix of $X$ with $Y[P[1 : k]]$ as a subsequence.
After $i^\star$ is initialized to $Q[k] - 1$, for each $j^\star$ from $P[k] + 1$ to $|Y|$, if there exists an index $i$ with $Q[k - 1] + 1 \leq i \leq i^\star$ such that $X[i] = Y[j^\star]$, then $i^\star$ is updated to the least such $i$, so that $L(X[Q[k - 1] + 1 : i^\star], Y[P[k - 1] + 1 : j^\star]) = 1$ and $X[1 : i^\star]$ is the shortest prefix of $X$ with $Y[P[1 : k - 1] \circ j^\star]$ as a subsequence.
Furthermore, after this update of $i^\star$, if $L(X[i^\star + 1 : |X|], Y[j^\star + 1 : |Y|]) = L(X, Y) - k$, then $\mathsf{k}^\star$ and $\mathsf{P}[k]$ are set to $k$ and $j^\star$ as $k^\star_P$ and $j^\star_P$, respectively, to complete the execution of \textsc{findBranch}; otherwise, $i^\star$ is decreased by one because $L(X[i + 1 : |X|], Y[j + 1 : |Y|]) < L(X, Y) - k$ for any indices $i$ with $i \geq i^\star$ and $j$ with $j \geq j^\star$.

Algorithm~\ref{algo findBranch} presents a pseudo-code of \textsc{findBranch}.
To provide quick access to $L(X[i^\star + 1 : |X|], Y[j^\star + 1 : |Y|])$, a technique similar to \textsc{firstLCS} is adopted.
That is, \textsc{findBranch} maintains variable $J$ so that $|J| = L(X[i^\star + 1 : |X|], Y)$ and for any index $l$ with $1 \leq l \leq |J|$, $J[l]$ is the greatest index $j$ with $L(X[i^\star + 1 : |X|], Y[j : |Y|]) = l$.
Procedure \textsc{decI} (lines 17 through 21) is used to decrease $i^\star$ by one and update $J$ accordingly.
Lines 1 through 4 construct $Q$ in a straightforward way.
The process for each $k$ is executed from $L(X, Y)$ to $1$ in descending order in the manner mentioned earlier by lines 7 through 16 

\begin{lemma}\label{lem findBranch}
For any element $P$ in $\mathcal{P}(X, Y)$, \textsc{findBranch} implemented as Algorithm~\ref{algo findBranch} sets $\mathsf{k}^\star$ and $\mathsf{P}[\mathsf{k}^\star]$ to $k^\star_P$ and $j^\star_P$, respectively, if $P$ is not the last element in $\mathcal{P}(X, Y)$, or sets $\mathsf{k}^\star$ to a dummy index, otherwise, in $O(|X||Y|)$ time and $O(L(X, Y))$ space.
\end{lemma}

\begin{proof}
Execution of lines 1 through 4 is done in $O(|X|)$ time.
Each execution of \textsc{decI} is done in $O(|Y|)$ time.
Hence, the process for each $k$ (lines 7 through 16) is executed in $O((Q[k] - Q[k - 1])|Y|)$ time, implying that \textsc{findBranch} runs in $O(|X||Y|)$ time.
It is easy to verify that \textsc{findBranch} uses $O(L(X, Y))$ space.
\end{proof}

\bibliographystyle{plain}
\bibliography{arxiv250522}

\end{document}